\documentclass{article}
\usepackage{authblk}
\usepackage{amsmath}
\usepackage{amsfonts}
\usepackage{amsthm}
\usepackage{amssymb}
\usepackage{algorithm,algpseudocode}
\usepackage{tikz}
\usepackage{array,booktabs}
\usepackage{thmtools,thm-restate}
\usepackage{subfig}
\DeclareMathOperator*{\expect}{\mathbb{E}}

\newtheorem{theorem}{Theorem}[section]
\newtheorem{lemma}[theorem]{Lemma}

\newcommand{\ground}{X}
\newcommand{\XX}{\mathcal{N}}
\newcommand{\A}{\mathcal{A}}
\newcommand{\B}{\mathcal{B}}

\newcommand{\D}{\mathcal{D}}
\newcommand{\mK}{\mathcal{K}}
\newcommand{\cS}{\mathcal{S}}
\newcommand{\cO}{\mathcal{O}}
\newcommand{\cR}{\mathcal{R}}
\newcommand{\Gaux}{G_{\!\A}}
\newcommand{\Haux}{H_{\!\A,\B}}
\newcommand{\subroutine}{\Call{Improve}{\Gaux}}

\newcolumntype{M}[1]{>{\centering\arraybackslash}m{#1}}

\title{Large Neighborhood Local Search for the Maximum Set Packing Problem}
\author{Maxim Sviridenko\thanks{M.I.Sviridenko@warwick.ac.uk, Work supported by EPSRC grant EP/J021814/1, FP7 Marie Curie Career Integration Grant and Royal Society Wolfson Research Merit Award. } }
\author{Justin Ward\thanks{J.D.Ward@dcs.warwick.ac.uk, Work supported by EPSRC grant EP/J021814/1. } }
\affil{Department of Computer Science, University of Warwick}
\date{}

\begin{document}
\tikzset{every node/.style={text width=1em,inner sep=0.5pt,text centered,font=\small}}
\maketitle
\begin{abstract}
In this paper we consider the classical maximum set packing problem where set cardinality is upper bounded by $k$. We show how to design a variant of a polynomial-time local search algorithm with performance guarantee $(k+2)/3$. This local search algorithm is a special case
of a more general procedure that allows to swap up to $\Theta(\log n)$ elements per iteration. We also design problem instances with locality gap $k/3$ even for a wide class of exponential time local search procedures, which can swap up to $cn$ elements for a constant $c$.  This shows that our analysis of this class of algorithms is almost tight.
\end{abstract}
\section{Introduction}
\label{sec:introduction}
In this paper, we consider the problem of maximum unweighted $k$-set packing.  In this problem, we are given a collection $\XX$ of $n$ distinct $k$-element subsets of some ground set $\ground$.
We say that two sets $A,B \in \XX$ \emph{conflict} if they share an element and call a collection of mutually non-conflicting sets from $\XX$ a \emph{packing}.  Then, the goal of the unweighted $k$-set packing problem is to find a packing $\A \subseteq \XX$ of maximum cardinality.  Here, we assume that each set has cardinality \emph{exactly} $k$.  This assumption is without loss of generality, since we can always add unique elements to each set of cardinality less than $k$ to obtain such an instance.

The maximum set packing problem is one the basic optimization problems. It received a significant amount of attention from researchers in the last few decades (see e.g. \cite{P}). It is known that a simple local search algorithm that starts with an arbitrary feasible solution and tries to add a constant number of sets to the current solution while removing a constant number of conflicting sets has performance guarantee arbitrarily close to $k/2$ \cite{HS}. It was also shown in \cite{HS} that the analysis of such an algorithm is tight, i.e. there are maximum set covering instances where the ratio between a locally optimal solution value and the globally optimal solution
value is arbitrarily close to $k/2$. 

Surprisingly, Halld\'orsson\cite{H} showed that if one increases the size of allowable swap to $\Theta(\log n)$ the performance guarantee can be shown to be at most $(k+2)/3$. Recently,  Cygan, Grandoni and Mastrolilli \cite{CGM} improved the guarantee for the same algorithm to $(k+1)/3$.  This performance guarantee is the best currently known for the maximum set packing problem. The obvious drawback of these algorithms is that it runs in time $O(n^{\log n})$ and therefore its running time not polynomial.

Both algorithms rely only on the subset of swaps of size $\Theta(\log n)$ to be able to prove their respective performance guarantees. The Halld\'orsson's swaps are particularly well structured and have a straightforward interpretation in the graph theoretic language.  In section \ref{sec:impl-search-canon} we employ techniques from fixed-parameter tractability to yield a procedure for finding well-structured improvements of size $O(\log n)$ in polynomial time.  Our algorithm is based on color coding technique introduced by Alon, Yuster, and Zwick \cite{Alon1995} and its extension by Fellows et al. \cite{Fellows2004}, and solves a dynamic program to locate an improvement if one exists.  Combining with Halld\'orsson's analysis, we obtain a polynomial time $\frac{k+2}{3}$-approximation algorithm.  In Section \ref{sec:lower-bound} we show that it is not possible to improve this result beyond $\frac{k}{3}$, even by choosing significantly larger improvements.  Specifically, we construct a family of instances in which the locality gap for a local search algorithm applying all improvements of size $t$ remains at least $\frac{k}{3}$ even when $t$ is allowed to grow linearly with $n$.  Our lower bound thus holds even for local search algorithms that are allowed to examine some exponential number of possible improvements at each stage.

\section{A Quasi-Polynomial Time Local Search Algorithm}
\label{sec:local-search-algor}
Let $\A$ be a packing.  We define an auxiliary multigraph $\Gaux$ whose vertices correspond to  sets in $\A$ and whose edges correspond to sets in $\XX \setminus \A$ that conflict with at most 2 sets in $\A$.  That is, $E(\Gaux)$ contains a separate edge $(S,T)$ for each set $X \in \XX \setminus \A$ that conflicts with exactly two sets $S$ and $T$ in $\A$, and a loop on $S$ for each set $X \in \XX$ that conflicts with exactly one set $S$ in $\A$.  In order to simplify our analysis, we additionally say that each set $X \in \A$ conflicts with itself, and place such a loop on each set of $\A$.  Note that $\Gaux$ contains $O(n)$ vertices and $O(n)$ edges, for any value of $\A$.

Our local search algorithm uses $\Gaux$ to search for improvements to the current solution $\A$.  Formally, we call a set $I$ of $t$ edges in $\Gaux$ a \emph{$t$-improvement} if $I$ covers at most $t - 1$ vertices of $\Gaux$ and the sets of $\XX \setminus \A$ corresponding to the edges in $I$ are mutually disjoint.

Note that if $I$ is a $t$-improvement for a packing $\A$, we can obtain a larger packing by removing the at most $t-1$ sets covered by $I$ from $\A$ and then adding the $t$ sets corresponding to the edges of $I$ to the result.  We limit our search for improvements in $\Gaux$ to those that exhibit the following particular form: an improvement is a \emph{canonical} improvement if it forms a connected graph containing two distinct cycles.  A general canonical improvement then comprises either 2 edge-disjoint cycles joined by a path, two edge-disjoint cycles that share a single vertex, or two distinct vertices joined by 3 edge-disjoint paths (see Figure \ref{fig:canonical}).\footnote{It can be shown that every $t$-improvement must contain a canonical $t$-improvement, and so we are not in fact restricting the search space at all by considering only canonical improvements.  However, this fact will not be necessary for our analysis.}  

Our algorithm, shown in Figure \ref{alg:1} proceeds by repeatedly calling the procedure $\subroutine$, which searches for a canonical $(4\log n + 1)$-improvement in the graph $\Gaux$.  Before searching for a canonical improvement, we first ensure that $\A$ is a maximal packing by greedily adding sets from $\XX \setminus \A$ to $\A$.  If $\subroutine$ returns an improvement $I$, then $I$ is applied to the current solution and the search continues. Otherwise, the current solution $\A$ is returned.  
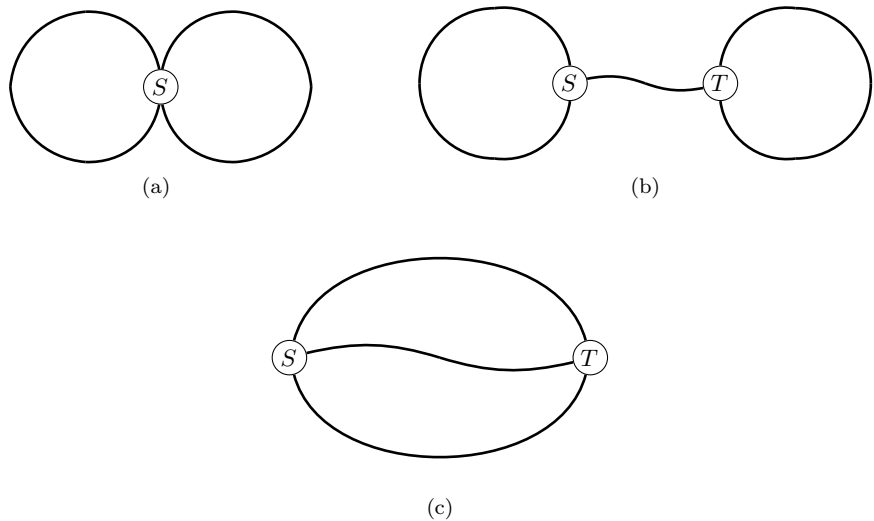
\begin{figure}[t]
\centering
\subfloat[]{
\label{fig:cannonb}
\centering
   \begin{tikzpicture}
     \node[draw,circle] (u) at (0,0) {$S$};
     \path[-,bend right=40,line width=1pt]
         (u) edge (-1,1)
         (-1,1) edge (-2,0)
         (-2,0) edge (-1,-1)
         (-1,-1) edge (u);
     \path[-,bend left=40,line width=1pt]
         (u) edge (1,1)
         (1,1) edge (2,0)
         (2,0) edge (1,-1)
         (1,-1) edge (u);
   \end{tikzpicture}
}
\hspace{3em}
\subfloat[]{
\label{fig:cannona}
\centering
   \begin{tikzpicture}
     \node[draw,circle] (u) at (0,0) {$S$};
     \node[draw,circle] (v) at (2,0) {$T$};
     \path[-,bend right=45,line width=1pt]
         (u) edge (-1,1)
         (-1,1) edge (-2,0)
         (-2,0) edge (-1,-1)
         (-1,-1) edge (u);
     \path[-,bend left=45,line width=1pt]
         (v) edge (3,1)
         (3,1) edge (4,0)
         (4,0) edge (3,-1)
         (3,-1) edge (v);
     \path[-,line width=1pt]
     (u) edge[bend left=15] (1,0)
     (1,0) edge[bend right=15] (v);
   \end{tikzpicture}
}
\newline
\subfloat[]{
\label{fig:cannonc}
\centering
\begin{tikzpicture}
     \node[draw,circle] (u) at (0,0) {$S$};
     \node[draw,circle] (v) at (4,0) {$T$};
     \path[-,line width=1pt]
     (u) edge[bend left=75] (v)
     (u) edge[bend left=15] (2,0)
     (2,0) edge[bend right = 15] (v)
     (u) edge[bend right=75] (v);
   \end{tikzpicture}
}
 \caption{Canonical Improvements}
   \label{fig:canonical}
\end{figure}

\begin{figure}
\begin{algorithmic}
\State  $\A \gets \emptyset$
\Loop
\ForAll{$S \in \XX \setminus \A$}
\If{$S$ does not conflict with any set of $\A$}
\State $\A \gets \A \cup \{ S \}$
\EndIf
\EndFor
\Statex
\State Construct the auxiliary graph $\Gaux$ for $\A$
\State $I \gets \subroutine$
\If{$I = \emptyset$}
\State \Return $\A$
\Else 
\State $\A \gets (\A \setminus V(I)) \cup E(I)$
\EndIf
\EndLoop
\end{algorithmic}
\caption{The General Local Search Procedure}
\label{alg:1}
\end{figure}

In Section \ref{sec:local-gap-algor}, we analyze the approximation performance of the local search algorithm under the assumption that $\subroutine$ always finds a canonical $(4 \log n + 1)$-improvement, whenever such an improvement exists.  In Section \ref{sec:impl-search-canon}, we provide such an implementation $\subroutine$ that runs in deterministic polynomial time.

\section{Locality Gap of the Algorithm}
\label{sec:local-gap-algor}

In this section we prove the following upper bound on the locality gap for our algorithm.   We consider an arbitrary instance $\XX$ of $k$-set packing, and let $\A$ be the packing in $\XX$ produced by our local search algorithm and $\B$ be any other packing in $\XX$.
\begin{theorem}
\label{thm:locality-gap}
$|\B| \le \frac{k+2}{3}|\A|$.
\end{theorem}

For the purpose of our analysis, we consider the subgraph $\Haux$ of $\Gaux$ consisting of only those edges of $\Gaux$ corresponding to sets in $\B$.  Then, every collection of edges in $\Haux$ is also present in $\Gaux$.  Moreover, because the edges of $\Haux$ all belong to the packing $\B$, any subset of them must be mutually disjoint.  Thus, we can assume that no collection of at most $4\log n + 1$ edges from $\Haux$ form any of the structures shown in Figure \ref{fig:canonical}.  Otherwise, the corresponding collection of edges in $\Gaux$ would form a canonical $(4\log n + 1)$-improvement.

In order to prove Theorem \ref{thm:locality-gap}, we make use of the following lemma of
Berman and F\"urer \cite{Berman1994}, which gives conditions under which the multigraph $\Haux$ must contain a canonical improvement.\footnote{Berman and F\"urer call structures of the form shown in Figure \ref{fig:canonical} ``binoculars.''  Here, we have rephrased their lemma in our own terminology.}  We provide Berman and F\"urer's proof in the appendix.
\begin{restatable}[Lemma 3.2 in \cite{Berman1994}]{lemma}{binocular}
\label{lem:binocular}
Assume that $|E| \ge \frac{p+1}{p}|V|$ in a multigraph $H=(V,E)$.  Then, $H$ contains a canonical improvement with at most $4p \log n - 1$ vertices.
\end{restatable}

It will also be necessary to bound the total number of loops in $\Haux$.  In order to do this, we shall consider a second auxiliary graph $\Haux'$ that is obtained from $\Haux$ in the following fashion:
\begin{restatable}{lemma}{removal}
\label{lem:induction}
Let $H=(V,E)$ be a multigraph and let $H'=(V',E')$ be obtained from $H$ by deleting all vertices of $H$ with loops on them and, for each edge with one endpoint incident to a deleted vertex, introducing a new loop on this edge's remaining vertex.  Let $t \ge 3$.  Then, if $H'$ contains a canonical $t$-improvement, $H$ contains a canonical $(t + 2)$-improvement.
\end{restatable}
A proof of Lemma \ref{lem:induction}, based on a sketch given by Halld\'orsson \cite{H}, appears in the appendix.

We now turn to the proof of Theorem \ref{thm:locality-gap}.  Every set in $\B$ must conflict with some set in $\A$, or else $\A$ would not be maximal.  We partition the sets of $\B$ into three collections of sets, depending on how many sets in $\A$ they conflict with.  Let $\B_1$,$\B_2$, and $\B_3$ be collections of those sets of $\B$ that conflict with, respectively, exactly 1, exactly 2, and 3 or more sets in $\A$ (note that each set of $\A \cap \B$ is counted in $\B_1$, since we have adopted the convention that such sets conflict with themselves).

Because each set in $\A$ contains at most $k$ elements and the sets in $\B$ are mutually disjoint, we have the inequality
\begin{equation}
  \label{eq:2}
  |B_1| + 2|B_2| + 3|B_{3}| \le k|A|.
\end{equation}

We now bound the size of $B_1$ and $B_2$.

Let $\A_1$ be the collection of sets from $\A$ that conflict with sets of $\B_1$.  Then, note that each set of $\B_1$ corresponds to a loop in $\Haux$ and the sets of $\A_1$ correspond to the vertices on which these loops occur.  Any vertex of $\Haux$ with two loops would form an improvement of the form shown in Figure \ref{fig:cannonb}.  Thus, each vertex in $\Haux$ has at most 1 loop and hence:
\begin{equation}
  \label{eq:6}
|\B_1| = |\A_1|.
\end{equation}

Now, we show that $|\B_2| \le 2|A \setminus \A_1|$.  By way of contradiction, suppose that $|\B_2| \ge 2|\A \setminus \A_1|$.  We construct an auxiliary graph $\Haux'$ from $\Haux$ as in Lemma \ref{lem:induction}.   The number of edges in this graph is exactly $|\B_2|$ and the number of vertices is exactly $|A \setminus \A_1|$.  Thus, if $|\B_2| \le 2|A \setminus \A_1|$, then from Lemma \ref{lem:binocular} (with $p = 1$), there is a canonical improvement in $\Haux'$ of size at most $4\log n - 1$.  But, from Lemma \ref{lem:induction} this means there must be a canonical improvement in $\Haux$ of size at most $4 \log n + 1$, contradicting the local optimality of $\A$.  Thus,
\begin{equation}
  \label{eq:8}
  |B_2| < 2|\A \setminus \A_1|
\end{equation}

Adding \eqref{eq:2}, twice \eqref{eq:6}, and \eqref{eq:8}, we obtain
\[
3|B_1| + 3|B_2| + 3|B_3| \le k|A| + 2|A_1| + 2|A \setminus A_1|,
\]
which implies that $3|B| \le (k + 2)|A|$.

\section{Finding Canonical Improvements}
\label{sec:impl-search-canon}

A na\"ive implementation of the local search algorithm described in Section \ref{sec:local-search-algor} would run in only quasi-polynomial time, since at each step there are $n^{\Omega(\log n)}$ possible improvements of size $t = 4\log n + 1$.  In contrast, we now show that it is possible to find a \emph{canonical} improvement of size $t$ in polynomial time whenever one exists.  

We first give a randomized algorithm, using the color coding approach of Alon, Yuster, and Zwick \cite{Alon1995}.  If some $t$-improvement exists, our algorithm finds it with polynomially small probability.  In Section \ref{sec:determ-algor}, we show how to use this algorithm to implement a local search algorithm that succeeds with high probability, and how to obtain to obtain a deterministic variant via derandomization.

We now describe the basic, randomized color coding algorithm.  Again, consider an arbitrary instance $\XX$ of $k$-set packing and let $\ground$ be the ground set of $\XX$.  Let $K$ be a collection of $kt$ colors.  We assign each element of $\ground$ a color from $K$ uniformly at random, and assign each $k$-set from $\XX$ the set of all its elements' colors.  We say that a collection of sets $\A \subseteq \XX$ is \emph{colorful} if no color appears twice amongst the sets of $\A$.  We note that if a collection of sets $\A$ is colorful, then $\A$ must form a packing, since no two sets in $\A$ can share an element.

We assign each edge of $\Gaux$ the same set of colors as its corresponding set in $\XX$, and, similarly, say that a collection of edges colorful if the corresponding collection of sets from $\XX$ is colorful.  Now, we consider a subgraph of $\Gaux$ made up of some set of at most $t$ edges $I$.  If this graph has the one of the forms shown in Figure \ref{fig:canonical} and $I$ is colorful, then $I$ must be a canonical $t$-improvement.  We now show that, although the converse does not hold, our random coloring makes any given canonical $t$-improvement colorful with probability only polynomially small in $n$.

Consider a canonical improvement $I$ of size $1 \le i \le t$.  The $i$ sets corresponding to the edges of $I$ must be disjoint and so consist of $ki$ separate elements.  The probability that $I$ is colorful is precisely the probability that all of these $ki$ elements are assigned distinct colors. This probability can be estimated as
\begin{equation}
\label{eq:7}
  \frac{\binom{kt}{ki}(ki)!}{(kt)^{ki}} = \frac{(kt)!}{(kt - ki)!(kt)^{ki}}\ge \frac{(kt)!}{(kt)^{kt}} > e^{-kt} = e^{-4k\log n - k} = e^{-k}n^{-8k},
\end{equation}
where in the last line, we have used the fact that $e^{\log n} = e^{\ln n\log e} = n^{\log e} < n^2$.

We now show how to use this random coloring to find canonical improvements in $\Gaux$.  Our approach is based on finding colorful paths and cycles and employs dynamic programming.

We give a dynamic program that, given a coloring for edges of $\Gaux$, as described above, finds a colorful path of length at most $t$ in $\Gaux$ between each pair of vertices $S$ and $T$, if such a path exists.  For each vertex $S$ and $T$ of $\Gaux$, each value $i \le t$, and each set $C$ of $ki$ colors from $K$, we have an entry $\D(S,T,i,C)$ that records whether or not there is some colorful path of length $i$ between $S$ and $T$ whose edges are colored with precisely those colors in $C$.  In our table, we explicitly include the case that $S = T$.

We compute the entries of $\D$ bottom-up in the following fashion.  We set $\D(S,T,0,C) = 0$ for all pairs of vertices $S$, $T$, and $C$.  Then, we compute the entries $\D(S,T,i,C)$ for $i > 0$ as follows.  We set $\D(S,T,i,C) = 1$, if there is some edge $(V,T)$ incident to vertex $T$ in $\Gaux$ such that $(V,T)$ is colored with a set of $k$ distinct colors $B \subseteq C$ and the entry $\D(S,V,i-1,C \setminus B) = 1$.  Otherwise, we set $\D(S,T,iC)=0$.

To determine if $\Gaux$ contains a colorful path of given length $i \le t$ from $S$ to $T$, we simply check whether $\D(S,T,i,C) = 1$ for some set of colors $C$.  Similarly, we can use our dynamic program to find colorful \emph{cycles} of length $j$ that include some given vertex $U$ by consulting $\D(U,U,j,C)$ for each set of colors $C$.  The actual path or cycle can then be found by backtracking through the table $\D$.  We note that while the cycles and paths found by this procedure are not necessarily simple, they are edge-disjoint.

For each value of $i$, there are at most $n^2\binom{kt}{ki}$ entries in $\D(S,T,i,C)$. To compute each such entry, we examine each edge $(V,T)$ incident to $T$, check if $(V,T)$ is colored with a set of $k$ colors $B \subseteq C$ and consult $\D(S,T,i,C\setminus B)$, all which can be accomplished in time $O(nki)$.  Thus, the total time to compute $\D$ up to $i = t$ is of order:
\begin{equation*}
 \sum_{i = 1}^t n^3ki \binom{kt}{ki} \le n^4kt2^{kt}
\end{equation*}

In order to find a canonical $t$-improvement, we first compute the table $\D$ up to $i = t$.  Then, we search for improvements of each kind shown in Figure \ref{fig:canonical} by enumerating over all choices of $S$ and $T$, and looking for an appropriate collection of cycles or paths involving these vertices that use mutually disjoint sets of colors. Specifically:
\begin{itemize}
\item  To find improvements of the form shown in Figure \ref{fig:cannonb},
we enumerate over all $n$ vertices $S$.  For all disjoint sets of $ka$ and $kb$ colors $C_a$ and $C_b$ with  $a + b \le t$, we check if $\D(S,S,a,C_a) = 1$ and $\D(S,S,b,C_b) = 1$.  This can be accomplished in time
\[
n \sum_{a = 1}^t\sum_{b = 1}^{t - a}2^{ka}2^{kb}kt = O(nkt^32^{kt})
\]
\item
To find improvements of the form shown in Figure \ref{fig:cannona} we enumerate over all distinct vertices $S$ and $T$.  For all disjoint sets of $ka$, $kb$, and $kc$ colors $C_a$,$C_b$,and $C_c$ with $|C_a| + |C_b| + |C_c| \le t$, we check if $\D(S,S,a,C_a) = 1$, $\D(T,T,b,C_b) = 1$, and $\D(S,T,c,C_c) = 1$.  This can be accomplished in time
\[
n^2 \sum_{a = 1}^t\sum_{b = 1}^{t - a}\sum_{c = 1}^{t - a - b}2^{ka}2^{kb}2^{kc}kt = O(n^2kt^42^{kt})
\]
\item
To find improvements of the form shown in Figure \ref{fig:cannonc} we again enumerate over all distinct vertices $S$ and $T$.  For all disjoint sets of $ka$, $kb$, and $kc$ colors $C_a$,$C_b$,and $C_c$ with $|C_a| + |C_b| + |C_c| \le t$, we check if $\D(S,T,a,C_a) = 1$, $\D(S,T,b,C_b) = 1$, and $\D(S,T,c,C_c) = 1$.  This can be accomplished in time
\[ \
n^2 \sum_{a = 1}^t\sum_{b = 1}^{t - a}\sum_{c = 1}^{t - a - b}2^{ka}2^{kb}2^{kc}kt = O(n^2kt^42^{kt})
\]
\end{itemize}

Thus, the total time spent searching for a canonical $t$-improvement is then at most:
\begin{equation*}
  O(n^2kt2^{kt} + nkt^32^{kt} + 2n^2kt^42^{kt}) = O(n^2kt^42^{kt}) = O(2^kk\cdot n^{4k+2}\log^4n).
\end{equation*}

\section{The Deterministic, Large Neighborhood Local Search Algorithm}
\label{sec:determ-algor}

The analysis of the local search algorithm in Section \ref{sec:local-gap-algor} supposed that every call to $\subroutine$ returns $\emptyset$ only when no canonical $t$-improvement exists in $\Gaux$.  Under this assumption, the algorithm is a $\frac{k+2}{3}$-approximation.  In contrast, the dynamic programming implementation 
given in Section \ref{sec:impl-search-canon} may fail to find a canonical improvement $I$ if the chosen random coloring does not make $I$ colorful.  As we have shown in \eqref{eq:7}, this can happen with probability at most $(1 - e^{-k}n^{-8k})$.

Suppose that we implement each call to $\subroutine$ by running the algorithm of Section \ref{sec:impl-search-canon} $cN = ce^{k}n^{8k}\ln n$ times, each with a different random coloring.  We now show that the resulting algorithm is a polynomial time $\frac{k+2}{3}$-approximation with high probability $1 - n^{1-c}$.

We note that each improvement found by our local search algorithm must increase the size of the packing $\A$, and so the algorithm makes at most $n$ calls to $\subroutine$.  We set $N = e^{k}n^{8k+1}\ln n$, and then implement each such call by repeating the color coding algorithm of Section \ref{sec:impl-search-canon} $cN$ times for some $c > 1$, each with an new random coloring. The probability that any given call $\subroutine$  succeeds in finding a canonical $t$-improvement when one exists is then at least:
\[
1 - (1 - e^{-k}n^{-8k})^{cN} \ge 1 - \exp\{e^{-k}n^{-8k}\cdot ce^kn^{8k}\ln n\} = 1 - n^{-c}.
\]
And so, the probability that all calls to $\subroutine$ satisfy the assumptions of Theorem \ref{thm:locality-gap} is at least:
\[
(1 - n^{-c})^n \ge 1 - n^{1-c}
\]

The resulting algorithm is therefore a $\frac{k+2}{3}$-approximation with high probability.  It requires at most $n$ calls to $\subroutine$, each requiring total time 
\[
O(cN \cdot 2^kk n^{4k+2}\log^4n) = O(c (2e)^kk n^{12k + 2}\log^nn\ln n) = cn^{O(k)}
\]

Using the general approach described by Alon, Yuster, and Zwick \cite{Alon1995}, we can in fact give a \emph{deterministic} implementation of $\subroutine$, which \emph{always} succeeds in finding a canonical $t$-improvement in $\Gaux$ if such an improvement exists.  Rather than choosing a coloring of the ground set $\ground$ at random, we use a collection $\mK$ of colorings (each of which is given as a mapping $\ground \to K$) with the property that every canonical $t$-improvement $\Gaux$ is colorful with respect to some coloring in $\mK$.  For this, it is sufficient to find a collection of $\mK$ of colorings such that for every set of at most $kt$ elements in $\ground$, there is some coloring in $\mK$ that assigns these $kt$ elements $kt$ distinct colors from $K$.  Then, we implement $\subroutine$ by running the dynamic programming algorithm of Section \ref{sec:determ-algor} on each such coloring, and returning the first improvement found.  Because every canonical $t$-improvement contains at most $kt$ distinct elements of the ground set, every such improvement must be made colorful by some coloring in $\mK$, and so $\subroutine$ will always find a canonical $t$-improvement if one exists.

We now show how to construct the desired collection of colorings $\mK$ by using a \emph{$kt$-perfect family} of hash functions from $\ground \to K$.  Briefly, a perfect hash function for a set $S \subseteq A$ is a mapping from $A$ to $B$ that is one-to-one on $S$.  A $p$-perfect family is then collection of perfect hash functions, one for each set $S \subseteq A$ of size at most $p$.  Building on work by Fredman, Koml\'os and Szemer\'edi \cite{FKS84} and Schmidt and Siegal \cite{Schmidt1990}, Alon and Naor show (in Theorem 3 of \cite{AN96}) how to explicitly construct a perfect hash function from $[m]$ to $[p]$ for some $S \subset [m]$ of size $p$ in time $\tilde{O}(p\log m)$.  This hash function is described in $O(p + \log p \cdot \log \log m)$ bits.  The maximum size of a $p$-perfect family of such functions is therefore $2^{O(p + \log p \cdot \log \log m)}$.  Moreover, the function can be evaluated in time $O(\log m/\log p)$.

Then, we can obtain a deterministic, polynomial time $\frac{k+2}{3}$ approximation as follows.  
Upon  receiving the set packing instance $\XX$ with ground set $\ground$, we compute a $kt$-perfect family $\mK$ of hash functions from $\ground$ to a set of $kt$ colors $K$.  Then, we implement each call to $\subroutine$ as described, by enumerating over the colorings in $\mK$.  We note that since each set in $\XX$ has size $k$, $|\ground| \le |\XX|k = nk$, so each improvement makes at most 
\[
2^{O(kt + \log kt \cdot \log \log kn)} = 2^{O(k\log n + \log (k \log n) \cdot \log\log kn)} = n^{O(k)}
\]
calls to the dynamic programming algorithm of Section \ref{sec:determ-algor} (one per coloring in $\mK$) and each of these calls takes time at most $n^{O(k)}$ (including the time to evaluate the coloring  on each element of the ground set).  Moreover, the initial construction of $\mK$ takes time at most $2^{kt}\tilde{O}(kt\log kn) = n^{O(k)}$.

\section{A Lower Bound}
\label{sec:lower-bound}
We now show that our analysis is almost tight.  Specifically, we show that the locality gap of $t$-local search is least $\frac{k}{3}$, even when $t$ is allowed to grow on the order of $n$.
\begin{theorem}
\label{thm:lower-bound-sublinear}
Let $c = \frac{9}{2e^5k}$ and suppose that $t \le cn$ for all sufficiently large $n$.  There, there exist 2 pairwise disjoint collections of $k$-sets $\cS$ and $\cO$ with $|\cS| = 3n$ and $|\cO| = kn$ such that any collection of $a \le t$ sets in $\cO$ conflict with at least $a$ sets in $\cS$.
\end{theorem}
In order to prove Theorem \ref{thm:lower-bound-sublinear} we make use of the following (random) construction.  Let $\ground$ be a ground set of $3kn$ elements, and consider a collection $\cS$ of $3n$ sets, each containing $k$ distinct elements of $\ground$.  We construct a random collection $\cR$ of $kn$ disjoint subsets of $\ground$, each containing 3 elements.

The number of distinct collections generated by this procedure is equal to the number of ways to partition the $3kn$ elements of $\ground$ into $kn$ disjoint 3-sets.  We define the quantity $\tau(m)$ to be the number of ways that $m$ elements can be divided into $m/3$ disjoint 3-sets:
\[
\tau(m) \triangleq \frac{m!}{(3!)^{m/3}(m/3)!}.
\]
In order to verify the above formula, consider the following procedure for generating a random partition. We first arrange the $m$ elements in some order and then make a 3-set from the elements at positions $3i$, $3i - 1$ and $3i-2$,  for each $i \in [\frac{m}{3}]$ (that is, we group each set of 3 consecutive elements in the ordering into a triple).  Now, we note that two permutations of the elements produce the same partition if they differ only in the ordering of the 3 elements within each of the $m/3$ triples or in the ordering of the $m/3$ triples themselves.  Thus, each partition occurs in exactly $(3!)^{m/3}(m/3)!$ of the $m!$ possible orderings.

The probability that any particular collection of $a$ disjoint 3-sets occurs in $\cR$ is given by
\[
p(a) \triangleq \frac{\tau(3kn - 3a)}{\tau(3kn)}.
\]
This is simply the number of ways to partition the remaining $3kn - 3a$ elements into $kn - a$ disjoint 3-sets, divided by the total number of possible partitions of all $3kn$ elements.

We say that a collection $\A$ of $a$ sets in $\cS$ is \emph{unstable} if there is some collection $\B$ of at least $a$ sets in $\cR$ that conflict with only those sets in $\A$.  Note that there is an improvement of size $a$ for $\cS$ only if there is some unstable collection $\A$ of size $a$ in $\cS$.\footnote{In fact, for an improvement to exist, there must be some collection $\B$ of $a+1$ such sets in $\cR$.  This stronger condition is unnecessary for our bound, however.}
We now derive an upper bound on the probability that our random construction of $\cR$ results in a given collection $\A$ in $\cS$ being unstable.
\begin{lemma}
\label{lem:good-sets}
A collection of $a$ sets in $\cS$ is unstable with probability less than $\frac{\binom{ka}{3da}\binom{kn}{da}}{\binom{3kn}{3da}}$.
\end{lemma}
\begin{proof}
A collection $\A$ of $a$ $k$-sets from $\cS$ is unstable precisely when there is a collection $\B$ of $a$ $3$-sets in $\cR$ that contain only those $k(a-1)$ elements appearing in the sets of $\A$.  There are $\binom{ka}{3a}$ ways to choose the $3a$ elements from which we construct $\B$.  For each such choice, there are $\tau(3a)$ possible ways to partition the elements into $3$-sets, each occurring with probability $p(3a) = \tau(3kn - 3a)/\tau(3n)$.  Applying the union bound, the probability that $\A$ is unstable is then at most:
\begin{align*}
\binom{ka}{3a}\tau(3a)\frac{\tau(3kn - 3a)}{\tau(3kn)} 
&= \binom{ka}{3a}\frac{(3a)!}{(3!)^{a}a!}\cdot\frac{(3(kn - a))!}{(3!)^{kn-a}(kn - a)!}\cdot\frac{(3!)^{kn}(kn)!}{(3kn)!} 
 \\
&= \binom{ka}{3a}\frac{(3a)!(3(kn - a))!}{(3kn)!}\frac{(kn)!}{(kn - a)!a!} \\
&= \frac{\binom{ka}{3a}\binom{kn}{a}}{\binom{3kn}{3a}} \\
\end{align*}
\qedhere
\end{proof}
\begin{proof}[Theorem \ref{thm:lower-bound-sublinear}]
Let $U_a$ be number of unstable collections of size $a$ in $\cS$, and consider $\expect[U_a]$.  There are precisely $\binom{3n}{a}$ such collections, and from Lemma \ref{lem:good-sets}, each occurs with probability less than $\frac{\binom{ka}{3a}\binom{kn}{a}}{\binom{3kn}{3a}}$. Thus:
\begin{equation}
  \label{eq:1}
\expect[U_a] < \frac{\binom{3n}{a}\binom{ka}{3a}\binom{kn}{a}}{\binom{3kn}{3a}}
\end{equation}
Applying the upper and lower bounds
\begin{displaymath}
\left(\frac{n}{k}\right)^k \le \binom{n}{k} \le \left(\frac{en}{k}\right)^k,
\end{displaymath}
in the numerator and denominator, respectively, of \eqref{eq:1}, we obtain the upper bound
\[
\frac{(e3n)^{a}}{a^a}\cdot \frac{(eka)^{3a}}{(3a)^{3a}}\cdot \frac{(ekn)^a}{a^a}\cdot \frac{(3a)^{3a}}{(3kn)^{3a}} =
\left(\frac{e^53^4k^4a^6n^2}{3^6k^3a^5n^3}\right)^a = 
 \left(\frac{e^5ka}{9n}\right)^a.
\]
Then, the expected number of unstable collections in $\cS$ of size \emph{at most} $t$ (and hence the expected number of $t$-improvements for $\cS$) is less than
\begin{align}
\sum_{a = 1}^{t}\expect[U_{a}] < \sum_{a = 1}^{t}\left(\frac{e^5ka}{9n}\right)^a.
\label{eq:3}
\end{align}
For all sufficiently large $n$, we have $a \le t \le cn$ and so
\begin{equation*}
\label{eq:5}
\sum_{a = 1}^{t}\left(\frac{e^5ka}{9n}\right)^a 
\le \sum_{a = 1}^{t}\left(\frac{e^5kcn}{9n}\right)^a
= \sum_{a = 1}^{t}\left(\frac{e^5k}{9}\frac{9}{2e^5k}\right)^a
= \sum_{a = 1}^{t}\left(\frac{1}{2}\right)^a
< 1.
\end{equation*}
Thus, there must exist some collection $\cO$ in the support of $\cR$ that creates no unstable collections of size at most $t$ in $\cS$.  Then, $\cO$ is a collection of pairwise disjoint sets of size $kn$ satisfying the conditions of the theorem.
\end{proof}

\section{Conclusion}

We have given a polynomial time $\frac{k+2}{3}$ approximation algorithm for the problem of $k$-set packing.  Our algorithm is based on a simple local search algorithm, but incorporates ideas from fixed parameter tractability to search large neighborhoods efficiently, allowing us to achieve an approximation guarantee exceeding the $k/2$ bound of Hurkens and Schrijver \cite{HS}.  In contrast, our lower bound of $k/3$ shows that local search algorithms considering still larger neighborhoods, including neighborhoods of exponential size, can yield only slight improvements.

An interesting direction for future research would be to close the gap between our $k/3$ lower bound and $\frac{k+2}{3}$ upper bound.  Recently, Cygan, Grandoni, and Mastrolilli \cite{CGM} have given a quasi-polynomial time local search algorithm attaining an approximation ratio of $(k+1)/3$.   Their analysis is also based on that of Berman and F\"urer \cite{Berman1994} and Halld\'orsson\cite{H}, but their algorithm requires searching for improvements with a more general structure than those that we consider, and it is unclear how to apply similar techniques as ours in this case.  Nevertheless, we conjecture that it is possible to attain an approximation ratio of $\frac{k+1}{3}$ in polynomial time, although this will likely require more sophisticated techniques than we consider here.

In contrast to all known positive results, the best known NP-hardness result for $k$-set packing is, due to Hazan, Safra, and Schwartz \cite{HSS06}, is only $O(k/\ln k)$.  A more general open problem is whether the gap between this result and algorithmic results can be narrowed.

Finally, we ask whether our results can be generalized to the independent set problem in $(k+1)$-claw free graphs.  Most known algorithms for $k$-set packing, including those given by Halld\'orsson \cite{H} and Cygan, Grandoni, and Mastrolilli \cite{CGM} generalized trivially to this setting.  However, this does not seem to be the case for the color coding approach that we employ, as it relies on the set packing representation of problem instances.

\section*{Acknowledgements} We would like to thank Oleg Pikhurko for extremely enlightening discussion on random graphs.

\appendix

\section{Appendix}
\label{sec:appendix}

Here we provide detailed proofs of the cited technical results.

\binocular*

\begin{proof}
Suppose that $H'$ is the smallest induced subgraph of $H$ that satisfies the condition 
\begin{equation}
\label{eq:11}
E(H') \ge \frac{p}{p+1}V(H').
\end{equation}
We shall show that $H'$ must contain a canonical improvement.  First, we note that $H'$ cannot contain any degree 1 vertices.  Otherwise, we could remove all such vertices to obtain a smaller graph satisfying~\eqref{eq:11}.  Moreover, any chain of degree 2 vertices in $H'$ has fewer than $p$ vertices.  Otherwise, we could remove this chain of $p$ vertices, together with the $p+1$ edges incident on them to obtain a smaller graph satisfying~\eqref{eq:11}.  We replace every chain of degree 2 vertices in $H'$ with a single edge connecting its the endpoints to obtain a graph $H_3$ with minimum degree 3.

Let $I_3$ be a minimal connected subgraph of $H_3$ with exactly 2 distinct cycles.  Then,  $I_3$ is a canonical improvement in $H_3$ and $|V(I_3)| \le |E(I_3)| + 1$.  By expanding each contracted chains of vertices in $I_3$, we obtain a connected subgraph of $H'$ that contains 2 distinct cycles.Then, $|V(I')| \le p(|V(I_3)| + 1)$.  Thus, to complete the proof of Lemma \ref{lem:binocular} it suffices to show that $H_3$ must contain a connected subgraph $I_3$ with exactly 2 distinct cycles and $|V(I_3)| \le 4\log n - 1$.

Let $n = |V|$ and note that $|V(H_3)| \le |V(H')| \le V(H) \le n$.  We first note that $H_3$ has maximum girth less than $2\log n$.  To prove this, simply construct a breadth first search tree rooted at some vertex in the graph.  There must be some vertex $v$ of distance less than $\log n$ from the root without 2 children.  Since $v$ has degree 3, it must have an edge to some previously visited vertex $u$ in the tree (where possibly $u = v$).  The paths from $u$ and $v$ to the root, together with the edge $(u,v)$ form a connected subgraph $C$ that has at most $2\log n$ vertices and contains both the root of the tree and a cycle.  Let $C$ be a minimal such subgraph.  We contract $C$ to a single vertex and call the resulting graph $H_3'$.  Then, $H_3'$ must also minimum degree 3 and at most $n$ vertices.  Repeating the argument we can find a minimal subgraph $C'$ in $H_3'$ with at most $2\log n$ vertices that contains both the root of $H_3'$ and a cycle.  Let $I_3$ be the induced subgraph of $H_3$ containing the vertices of $C$ and $C'$.  Then, $I_3$ has at most $4\log n - 1$ vertices, and is a connected subgraph of $H_3$ containing exactly 2 distinct cycles.
\end{proof}

\removal*

\begin{proof}
Our argument is based on a sketch given by Halld\'orsson \cite{H}.  In the interest of completeness, we present a more detailed argument here.

Note that any canonical $t$ improvement in $H'$ that does not contain a loop is also present in $H$.   Consider, then, a canonical $t$-improvement $I$ in $H'$ which contains either one or two loops (i.e. an improvement of the form \ref{fig:cannonb} and \ref{fig:cannona}, where one or both of the cycles are loops).  A loop on vertex $v$ in $H'$ corresponds to an edge $(u,v)$ in $H$, where $v$ has a loop.  The improvement $I$ in $H'$ must have 2 cycles joined by either a path or a single vertex.  Figure \ref{fig:secondary-reduction} illustrates all of the possible configurations for $I$ in $H'$ and the related canonical improvements in $H$, which we now show must exist.

If exactly one of these cycles is a loop on some vertex $v$, then we must have a path (possibly of length 0) joining $v$ to another cycle $J$ in $H'$.  This path and $J$ are also present in $H$.  Additionally, in $H$ we must have an edge $(u,v)$ and a loop on $u$.  Thus, the loop on $u$, together with the edge $u,v$, the cycle $J$ and the path connecting $v$ to $J$ form a canonical improvement with only one more edge than $I$.

Now, suppose that both of these cycles are loops on some vertices $v_1$ and $v_2$ (where possibly $v_1 = v_2$).  If the corresponding edges $(v_1,u_1)$ and $v_2,u_2$ in $H$ have distinct endpoints $u_1 \neq u_2$, then the two loops on $u_1$ and $u_2$ together with these edges and the path (of length 0, in the case that $v_1 = v_2$) joining $v_1$ and $v_2$ form a canonical improvement.  If $u_1 = u_2$, then the edges $(v_1,u_1)$ and $(v_2,u_2)$, together with the path (again, of length 0 if $v_1 = v_2$) from $v_1$ to $v_2$ forms a cycle.  This, together with the loop on the vertex $u_1 = u_2$, forms a canonical improvement.  In both cases, the canonical improvement in $H$ has only 2 more edges than $I$.
\end{proof}

\begin{figure}
\tikzset{every loop/.style={min distance=5mm,looseness=10}}
\centering
\begin{tabular}{M{0.4\linewidth}M{0.4\linewidth}}
\toprule
Improvement in $H'$
&
Corresponding Improvement in $H$
\\ \midrule
   \begin{tikzpicture}
     \node[draw,circle] (u) at (0,0) {$S$};
     \path[-,bend right=45,line width=1pt]
         (u) edge (-1,1)
         (-1,1) edge (-2,0)
         (-2,0) edge (-1,-1)
         (-1,-1) edge (u);
     \path[-]
         (u) edge[loop right] (u);
   \end{tikzpicture}
&
   \begin{tikzpicture}
     \node[draw,circle] (u) at (0,0) {$S$};
     \node[draw,circle] (z) at (1,0) {$X$};
     \path[-,bend right=45,line width=1pt]
         (u) edge (-1,1)
         (-1,1) edge (-2,0)
         (-2,0) edge (-1,-1)
         (-1,-1) edge (u);
     \path[-]
         (z) edge[loop right] (z);
     \path[-]
     (u) edge (z);
   \end{tikzpicture}
\\   \midrule
   \begin{tikzpicture}
     \node[draw,circle] (u) at (0,0) {$S$};
     \path[-]
         (u) edge[loop left] (u)
         (u) edge[loop right] (u);
   \end{tikzpicture}
&
   \begin{tikzpicture}
     \node[draw,circle] (u) at (0,0) {$S$};
     \node[draw,circle] (z1) at (-1,0) {$X$};
     \node[draw,circle] (z2) at (1,0) {$Y$};
     \path[-]
         (z1) edge[loop left] (z1)
         (z2) edge[loop right] (z2);
     \path[-]
         (z1) edge (u)
         (u) edge (z2);
     \node[draw,circle] (u) at (0,-1.5) {$S$};
     \node[draw,circle] (z1) at (0,-3) {$X$};
     \path[-]
         (z1) edge[loop below] (z1);
     \path[-]
         (z1) edge[bend right=30] (u)
         (z1) edge[bend left=30] (u);
       \end{tikzpicture}
\\
\midrule
   \begin{tikzpicture}
     \node[draw,circle] (u) at (0,0) {$S$};
     \node[draw,circle] (v) at (2,0) {$T$};
     \path[-,bend right=45,line width=1pt]
         (u) edge (-1,1)
         (-1,1) edge (-2,0)
         (-2,0) edge (-1,-1)
         (-1,-1) edge (u);
     \path[-]
         (v) edge[loop right] (v);
     \path[-,line width=1pt]
     (u) edge[bend left=15] (1,0)
     (1,0) edge[bend right=15] (v);
   \end{tikzpicture}
&
   \begin{tikzpicture}
     \node[draw,circle] (u) at (0,0) {$S$};
     \node[draw,circle] (v) at (2,0) {$T$};
     \node[draw,circle] (z) at (3,0) {$X$};
     \path[-,bend right=45,line width=1pt]
         (u) edge (-1,1)
         (-1,1) edge (-2,0)
         (-2,0) edge (-1,-1)
         (-1,-1) edge (u);
     \path[-]
         (z) edge[loop right] (z);
     \path[-,line width=1pt]
     (u) edge[bend left=15] (1,0)
     (1,0) edge[bend right=15] (v);
    \path[-]
    (v) edge (z);
   \end{tikzpicture}
\\ \midrule
   \begin{tikzpicture}
     \node[draw,circle] (u) at (0,0) {$S$};
     \node[draw,circle] (v) at (2,0) {$T$};
     \path[-]
         (u) edge[loop left] (u)
         (v) edge[loop right] (v);
     \path[-,,line width=1pt]
     (u) edge[bend left=15] (1,0)
     (1,0) edge[bend right=15] (v);
   \end{tikzpicture}
 &
 \begin{tikzpicture}
     \node[draw,circle] (u) at (0,0) {$S$};
     \node[draw,circle] (v) at (2,0) {$T$};
     \node[draw,circle] (z1) at (-1,0) {$X$};
     \node[draw,circle] (z2) at (3,0) {$Y$};
     \path[-]
         (z1) edge[loop left] (z1)
         (z2) edge[loop right] (z2)
         (z1) edge (u)
         (v) edge (z2);
     \path[-,line width=1pt]
         (u) edge[bend left=15] (1,0)
         (1,0) edge[bend right=15] (v);
     \node[draw,circle] (u) at (0,-1.5) {$S$};
     \node[draw,circle] (v) at (2,-1.5) {$T$};
     \node[draw,circle] (z1) at (1,-2.5) {$X$};
     \path[-]
         (z1) edge[loop below] (z1)
         (z1) edge (u)
         (v) edge (z1);
     \path[-,line width=1pt]
         (u) edge[bend left=15] (1,-1.5)
         (1,-1.5) edge[bend right=15] (v);
       \end{tikzpicture}
\\ \bottomrule
 \end{tabular}
   \caption{Canonical improvements in $H'$ and corresponding improvements in $H$.}
 \label{fig:secondary-reduction}
\end{figure}
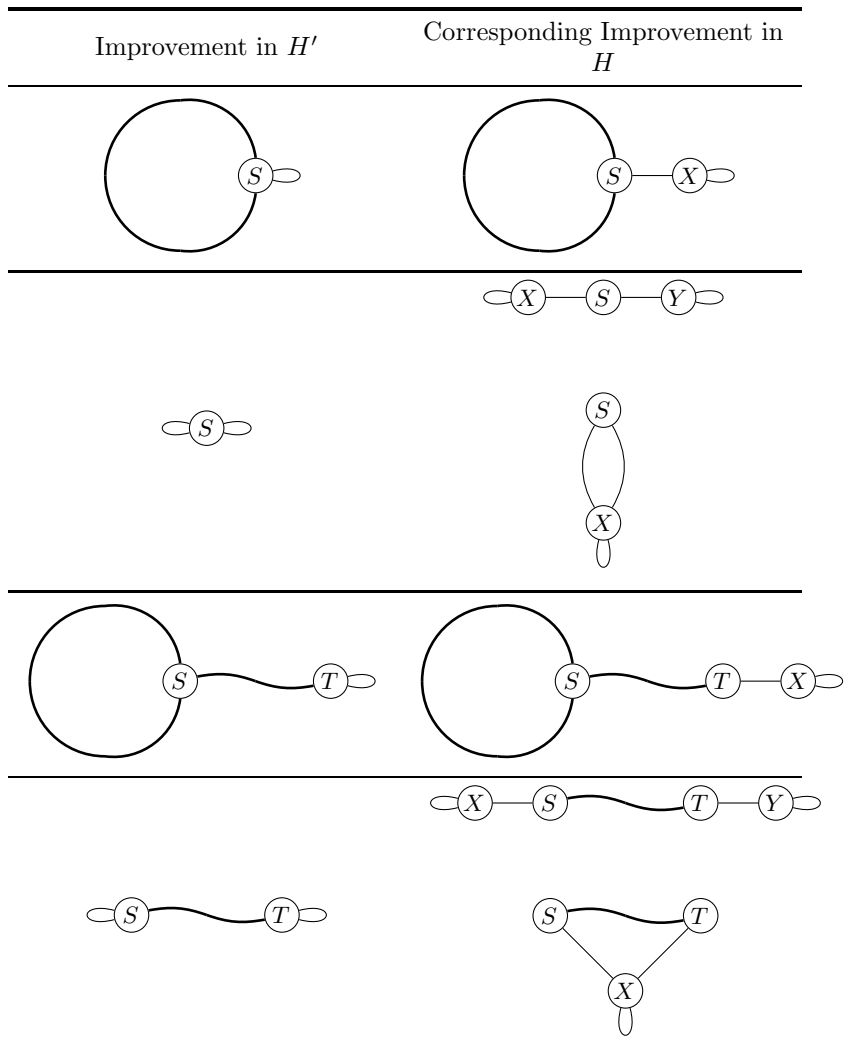
\end{document}